\documentclass{llncs}

\usepackage{amsmath}
\usepackage{amssymb}
\usepackage{color}
\usepackage{graphicx}
\usepackage{cancel}

\def \lket {|}
\def \rket {\rangle}
\def \lbra {\langle}
\def \rbra {|}

\newcommand{\ket}[1]{\lket #1\rket}
\newcommand{\bra}[1]{\lbra #1\rbra}
\newcommand{\comment}[1]{}
\renewcommand{\labelitemi}{$\circ$}

\newtheorem{Theorem}{Theorem}
\newtheorem{Lemma}{Lemma}
\newtheorem{Claim}{Claim}

\DeclareMathOperator{\arccot}{arccot}

\begin{document}


\title{Grover's search with faults on some marked elements}
\author{Dmitry Kravchenko, Nikolajs Nahimovs, Alexander Rivosh\thanks{This research was supported by EU FP7 project QALGO (Dmitry Kravchenko, Nikolajs Nahimovs) and ERC project MQC (Alexander Rivosh)}}
\institute{Faculty of Computing, University of Latvia}

\maketitle


\begin{abstract}

Grover's algorithm is a quantum query algorithm solving the unstructured search problem of size $N$ using $O(\sqrt{N})$ queries.
It provides a significant speed-up over any classical algorithm \cite{Gro96}.

The running time of the algorithm, however, is very sensitive to errors in queries. It is known that if query may fail (report all marked elements as unmarked) the algorithm needs $\Omega(N)$ queries to find a marked element \cite{RS08}. \cite{AB+13} have proved the same result for the model where each marked element has its own probability to be reported as unmarked. 

We study the behavior of Grover's algorithm in the model where the search space contains both faulty and non-faulty marked elements.
We show that in this setting it is indeed possible to find one of non-faulty marked items in $O(\sqrt{N})$ queries.

We also analyze the limiting behavior of the algorithm for a large number of steps and show the existence and the structure of limiting state $\rho_{lim}$.

\end{abstract}


\section{Introduction}

Grover's algorithm is a quantum query algorithm solving the unstructured search problem of size $N$ using $O(\sqrt{N})$ queries.
It is known that any deterministic or randomized algorithm needs linear time (number of queries) to solve the above problem. 
Thus, Grover's algorithm provides a significant speed-up over any classical algorithm.

The running time of the algorithm (number of queries), however, is very sensitive to errors in queries.
Regev and Schiff \cite{RS08} have shown that if query has a small probability of failing (reporting that \textit{none} of the elements are marked), then quantum speed-up disappears:
no quantum algorithm can be faster than a classical exhaustive search by more than a constant factor.
Ambainis et al. \cite{AB+13} have studied Grover's algorithm in the model where each marked element has its own probability to be reported as unmarked, independent of probabilities of other marked elements. 
Similarly to the result of \cite{RS08}, they have shown that if all marked elements are faulty (have non-zero probability of failure) then the algorithm needs $\Omega(N)$ queries to find a marked element.

Although, technically the model of \cite{AB+13} allows one non-faulty marked element\footnote{The limitation of at most one non-faulty marked element comes from the probability independence assumption -- two or more marked elements with zero error probability of failure would not be independent.} (element with zero probability of failure) this case was not included into the analysis.

We study the behavior of the algorithm in the model where the search space contains both faulty and non-faulty marked items.
Specifically, we focus on the case where the search space contains multiple non-faulty and one faulty marked element. 
We analyze the effect of a fault on the state of the algorithm and show that in this setting it is indeed possible to find one of non-faulty marked elements in $O(\sqrt{N})$ queries.
Up to the best our knowledge, this is the first demonstrgggggation of query fault modes which can be tolerated by the Grover's algorithm.

We also analyze the limiting behavior of Grover's algorithm for a large number of steps and show the existence and the structure of limiting state $\rho_{lim}$. 
The limiting state is a mixture of a faulty marked element $\ket{i}$ with probability $\frac{1}{3}$, the uniform superposition of all non-faulty marked elements with probability $\frac{1}{3}$ and the uniform superposition of all non-marked elements with probability $\frac{1}{3}$.
Using the approach of \cite{AB+13} one can show that convergence time is $O(N)$.

The contrast between this result and the algorithm finding a marked element in $O(\sqrt{N})$ steps can be explained by the probability of finding a marked element oscillating between $\Omega(1)$ and $o(1)$ until these oscillations finally die off after O($N$) steps.


\section{Technical preliminaries}

We use the standard notions of quantum states, density matrices etc., as described in \cite{KLM07}.


\subsection*{Grover's algorithm \cite{Gro96}}

Suppose we have an unstructured search space of size $N$.
Grover's algorithm starts with a starting state 
$\ket{\psi_{0}}= \frac{1}{\sqrt{N}} \sum_{i=1}^N \ket{i}$. Each step of the algorithm consists of two transformations: $Q$ and $D$. Here, $Q$ is a query to a black box defined by
\begin{itemize}
\item
$Q \ket{i} = - \ket{i}$ if $i$ is a marked element;
\item
$Q \ket{i} = \ket{i}$ if $i$ is not a marked element.
\end{itemize}
$D$ is the diffusion transformation described by the following $N \times N$ matrix:
\[ D = \left( \begin{array}{cccc} 
-1 + \frac{2}{N} & \frac{2}{N} & \ldots & \frac{2}{N} \\
\frac{2}{N} & -1 + \frac{2}{N} & \ldots & \frac{2}{N} \\
\ldots & \ldots & \ldots & \ldots \\
\frac{2}{N} & \frac{2}{N} & \ldots & -1+ \frac{2}{N} \end{array} \right) .\]
We refer to $\ket{\psi_t}= (DQ)^t \ket{\psi_{0}}$ as the state of Grover's algorithm after $t$ time steps.

Grover's algorithm has been analyzed in detail and many facts about it are known \cite{Amb04}.
If there is one marked element $i$, the probability of finding it by measuring $\ket{\psi_t}$ 
reaches $1-o(1)$ for $t=O(\sqrt{N})$. If there are $k$ marked elements, the probability of finding one of them
by measuring $\ket{\psi_t}$ reaches $1-o(1)$ for $t=O(\sqrt{N/k})$. 


\subsection*{Spherical trigonometry}

Spherical trigonometry is a branch of geometry which deals with the relationships between trigonometric functions of the sides and angles of the spherical polygons.
Trigonometry on a sphere differs from the traditional planar trigonometry.
For example, all distances are measured as angular distances.

In the context of this paper we need only a few basic formula for right spherical triangles.
Let $ABC$ be a right spherical triangle with a right angle $C$.
Then the following following set of rules (known as  Napier's rules) applies:

\begin{tabular}{p{13em}@{\qquad\qquad}p{13em}}
{\begin{align}
\cos{c} & = \cos{a}\cos{b} \tag{R1}\label{eq:R1} \\
\sin{a} & = \sin{A}\sin{c} \tag{R2}\label{eq:R2} \\
\sin{b} & = \sin{B}\sin{c} \tag{R3}\label{eq:R3} \\
\tan{a} & = \tan{A}\sin{b} \tag{R4}\label{eq:R4} \\
\tan{b} & = \tan{B}\sin{a} \tag{R5}\label{eq:R5}
\end{align}}
&
{\begin{align}
\tan{b} & = \cos{A}\tan{c} \tag{R6}\label{eq:R6} \\
\tan{a} & = \cos{B}\tan{c} \tag{R7}\label{eq:R7} \\
\cos{A} & = \sin{B}\cos{a} \tag{R8}\label{eq:R8} \\
\cos{B} & = \sin{A}\cos{b} \tag{R9}\label{eq:R9} \\
\cos{c} & = \cot{A}\cot{B} \tag{R10}\label{eq:R10}
\end{align}}
\end{tabular}
\begin{figure}[h]
\centering
\includegraphics[scale=0.5]{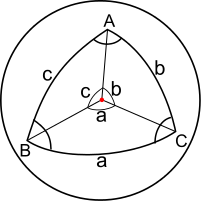}
\caption{Spherical trigonometry basic triangle.}
\label{fig:1}
\end{figure}

For more detailed introduction into spherical trigonometry see \cite{Tod86}.


\section{Model and results}


\subsection*{Error model}
 
Suppose we have a search space of size $N$ containing $k$ marked elements $i_1, i_2, \ldots, i_k$.
First $k - 1$ marked elements are {\em non-faulty} -- the query always returns them as marked.
Last marked element is {\em faulty} -- the query might return it as unmarked.

More formally, on each step, instead of the correct query $Q$, we apply 
a faulty query $Q'$ defined as follows:

\begin{itemize}
\item
$Q' \ket{i_k} = \ket{i_k}$ with probability $p$;
\item
$Q' \ket{i_k} = - \ket{i_k}$ with probability $1 - p$;
\item
$Q' \ket{j} = Q \ket{j}$ if $j \neq i_k$.
\end{itemize}


\subsection*{Summary of results}

First we show that if there is at least one non-faulty marked element, then it is still possible to find a non-faulty marked element in $O(\sqrt{N})$ queries with $\Theta(1)$ probability.

\begin{Theorem}
\label{thm:search}
Let $k \geq 3$, then we can choose $t = O(\sqrt{N/k})$ so that, if we run Grover's algorithm for $t$ steps and measure the final state, the probability of finding a marked element is at least $\cos^{2}\frac{\pi}{8} = 0.85\ldots$.

For $k=2$, the probability of finding a marked element is at least 
$\cos^{2}\frac{\pi}{8} = 0.85\ldots$ under a promise that at most one fault occurs and at least $0.74\ldots$ in the general case.
\end{Theorem}

\noindent
We conjecture that, for $k=2$, the probability is at least $0.85\ldots$ even in the general case.

Second, we analyze the limiting behavior of the algorithm for a large number of steps and show the existence and structure of limiting state $\rho_{lim}$.

\begin{Theorem}
\label{thm:limiting_behaviour}
Let $\rho_t$ be the density matrix of state of Grover's algorithm with a faulty oracle $Q'$ after $t$ queries.
Then, the sequence $\rho_1, \rho_2, \ldots$ converges to 
$$
\rho_{lim} = \frac{1}{3} \ket{\psi_{+}}\bra{\psi_{+}} + \frac{1}{3} \ket{i_k}\bra{i_k} + \frac{1}{3} \ket{\psi_{-}}\bra{\psi_{-}} 
$$
where $\ket{\psi_{+}} = \frac{1}{\sqrt{k-1}} \sum_{j=1}^{k-1}\ket{i_j}$ is the uniform superposition over all non-faulty marked elements and $\ket{\psi_{-}} = \frac{1}{\sqrt{N-k}} \sum_{i \neq i_j} \ket{i}$ is the uniform superposition over all non-marked elements.
\end{Theorem}


\section{Search}

In this section we analyze the evolution of the state of Grover's algorithm in presence of multiple non-faulty and one faulty marked item. 
First we review the original Grover's algorithm, then we describe the effect of faults on the state of the algorithm.
We derive upper bounds on the effect of faults and
provide a modification of the original Grover's search algorithm which finds one of non-faulty marked items with $\Theta(1)$ probability in $O(\sqrt{N})$ queries.

\subsection{No faulty marked items}

Let us first consider the very basic search problem of Grover's algorithm.
Namely, we have $N$ items among which $k$ are marked\footnote{
It is usually considered that $k \ll N$,
as for $\frac{k}{N}\ge\lambda$ with sufficiently large $\lambda$
the search problem can be trivially solved by a probabilistic algorithm in time $O\left(\lambda^{-1}\right)$.
}.

Operator $D$ is symmetric w.r.t. permutations of amplitudes of all items,
and operator $Q$ is symmetric w.r.t. permutations of amplitudes of marked items,
as well as permutations of amplitudes of non-marked items.
So, on any step $t$ amplitudes of all marked items are equal to each other and amplitudes of all non-marked items are equal to each other.
Thus, we can represent $\ket{\psi_t}$ as:
\begin{displaymath}
\ket{\psi_t} = \sum_{i \in U} \alpha_t \ket{i} + \sum_{j \in M} \beta_t \ket{j},
\end{displaymath}
where $U$ stands for the set of indexes of non-marked items and $M$ stands for the set of indexes of $k$ marked items.
$\alpha_t$ and $\beta_t$ denote the amplitudes of respectively a non-marked item and a marked item on step $t$.
At each step of the algorithm we shall take care of two numbers only:
\begin{equation}
\label{eq:alphabeta}
\alpha_t\sqrt{N-k}\textrm{\qquad and \qquad}\beta_t\sqrt{k}.
\end{equation}
Since $\ket{\psi_t}$ is a unit vector, we have
\begin{displaymath}
\sum_{i \in U} \alpha_t^2 + \sum_{j \in M} \beta_t^2 = 1.
\end{displaymath}
Thus, values \eqref{eq:alphabeta} meet the equality \quad$\left(\vphantom{\beta_t\sqrt{k}}\alpha_t\sqrt{N-k}\right)^2+\left(\beta_t\sqrt{k}\right)^2=1$\quad
and correspond to a point on the unit circle.

Initially all amplitudes are equal, so $\alpha=\beta=\frac{1}{\sqrt{N}}$, and the numbers \eqref{eq:alphabeta} are
\begin{equation}
\label{eq:alphabeta0}
\alpha_0\sqrt{N-k}=\frac{\sqrt{N-k}}{\sqrt{N}}\textrm{\qquad and \qquad}\beta_0\sqrt{k}=\frac{\sqrt{k}}{\sqrt{N}}.
\end{equation}
During the first step of the algorithm operator $D$ does not change amplitudes of the state $\ket{\psi_0}$,
and operator $Q$ negates amplitudes of all marked items: $\beta_1=-\beta_0=-\frac{1}{\sqrt{N}}$.
So the numbers \eqref{eq:alphabeta} are
\begin{equation}
\label{eq:alphabeta1}
\alpha_1\sqrt{N-k}=\frac{\sqrt{N-k}}{\sqrt{N}}\textrm{\qquad and \qquad}\beta_1\sqrt{k}=-\frac{\sqrt{k}}{\sqrt{N}}.
\end{equation}

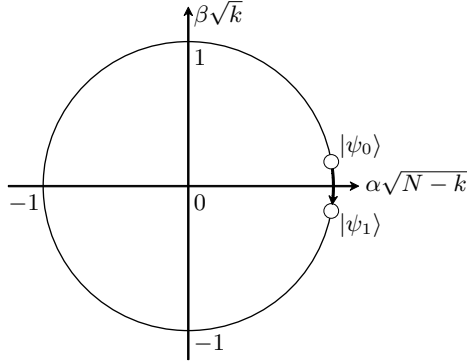
\begin{figure}[h]
\centering
\begin{picture}(150,150)
\put(0,0){\includegraphics[width=150pt,height=150pt]{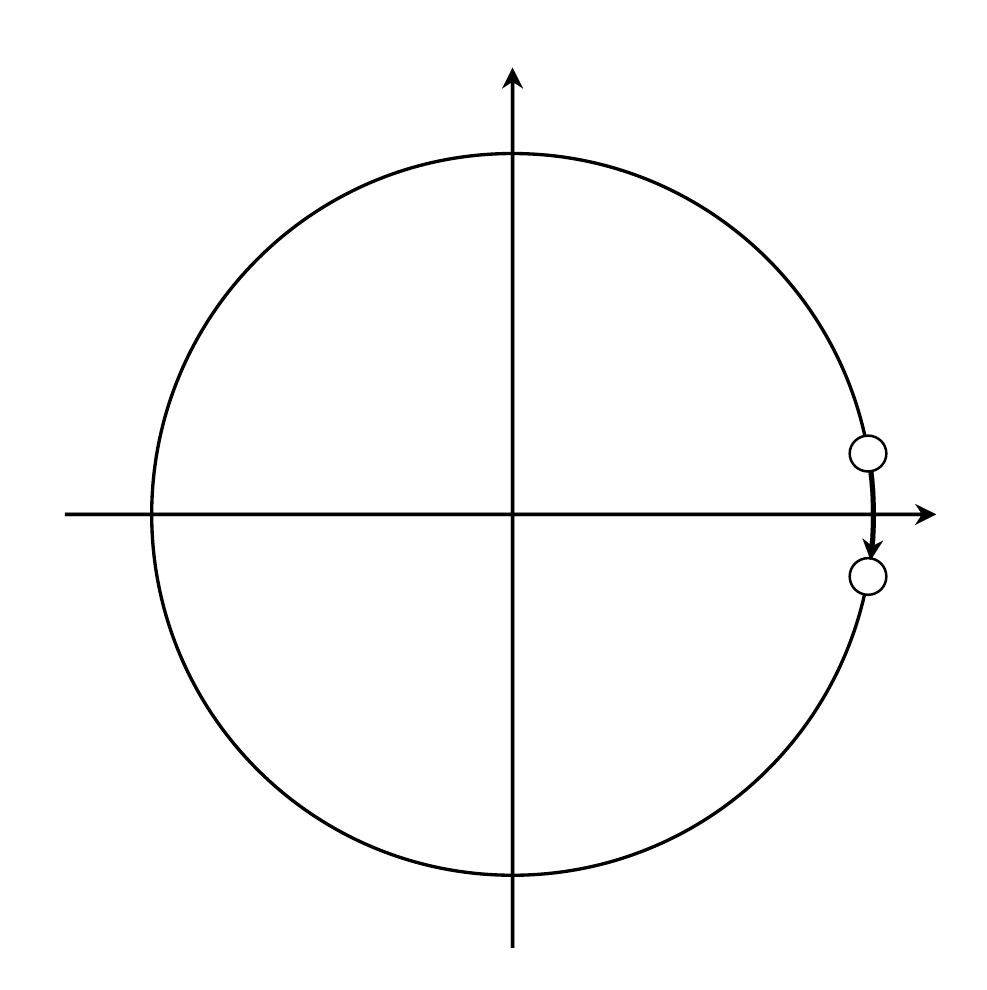}}
\put(10,63){\hbox{$-1$}}
\put(80,10){\hbox{$-1$}}
\put(80,63){\hbox{$0$}}
\put(80,118){\hbox{$1$}}
\put(145,70){\hbox{$\alpha\sqrt{N-k}$}}
\put(80,134){\hbox{$\beta\sqrt{k}$}}
\put(135,85){\hbox{$\ket{\psi_0}$}}
\put(135,55){\hbox{$\ket{\psi_1}$}}
\end{picture}
\caption{The first step of Grover's algorithm}
\label{fig:fig1}
\end{figure}

According to \eqref{eq:alphabeta0} and \eqref{eq:alphabeta1},
transformation $\ket{\psi_0} \xrightarrow{DQ} \ket{\psi_1}$
can be represented on the unit circle as shown on Figure~\ref{fig:fig1}.
As before, we assume $k \ll N$, so that $\sqrt{k} \ll \sqrt{N-k}$,
and the angle between $\ket{\psi_0}$ and $\ket{\psi_1}$ is
\begin{equation}
\label{eq:rotation}
2\arcsin{\left(\frac{\sqrt{k}}{\sqrt{N}}\right)} \approx 2\frac{\sqrt{k}}{\sqrt{N}}
\end{equation}
(this approximation holds for small-valued $\frac{\sqrt{k}}{\sqrt{N}}$).

Similarly, all further applications of operator $DQ$ are nothing but clockwise rotations by angle
$\sim2\frac{\sqrt{k}}{\sqrt{N}}$. After $\sim\frac{\pi/2}{2\sqrt{k}/\sqrt{N}}=\frac{0.785\ldots\sqrt{N}}{\sqrt{k}}$
such rotations the resulting state $\ket{\psi_{\lfloor{0.785\ldots\sqrt{N/k}}\rfloor}}$
reaches the neighborhood of the point $\left(0,-1\right)$.
Measuring $\ket{\psi_{\lfloor{0.785\ldots\sqrt{N/k}}\rfloor}}$ results in getting index of a marked item, with probability almost $1$.


\subsection{One Faulty Marked Item}

Let us now consider the case with
\begin{itemize}
\renewcommand{\labelitemi}{$\bullet$}
\item $N-k$ non-marked items,
\item $k-1$ marked items, and
\item $1$ \textit{faulty} marked item.
\end{itemize}
To simplify the analysis we shall interpret the step of the algorithm as consequent application of three operators:
ordinary diffusion $D$ and ordinary query $Q$, and -- with probability $\epsilon$ -- error $E$, which negates back the amplitude of the faulty marked item.

As the operation $E$ is probabilistic one must deal not with a pure state $\ket{\psi_t}$, but with a mixed state $\rho_t$ (probabilistic mixture of pure states).
We shall denote components of the mixture after $t$ steps as
$\ket{\psi^w_t}$, where $w \in \left\{0,1\right\}^t$ stands for the sequence of $t$ events:
$0$ -- the query has negated the amplitude of the faulty marked item ($DQ$),
and $1$ -- the query has left that amplitude of the faulty marked item unchanged ($DQE$).
So the mixture $\rho_t$ looks as follows:

\begin{displaymath}
\rho_t = \sum_{w \in \left\{0,1\right\}^t}{\epsilon^{\left|w\right|}\left(1-\epsilon\right)^{t-\left|w\right|}}\ket{\psi^w_t}\bra{\psi^w_t}.
\end{displaymath}

Transformations $D$, $Q$ and $E$ are symmetric
w.r.t. permutations of amplitudes of non-faulty marked items,
as well as permutation of amplitudes of non-marked items.
So, in any state $\ket{\psi^w_t}$ of the mixture $\ket{\psi^*_t}$,
amplitudes of all non-faulty marked items are equal to each other
and amplitudes of all non-marked items are equal to each other.
Thus, we can represent $\ket{\psi^w_t}$ as:
\begin{equation}
\label{eq:psiwt}
\ket{\psi^w_t} = \sum_{i \in U} \alpha^w_t \ket{i} + \sum_{\substack{j \in M,\\ j \neq i_k}} \beta^w_t \ket{j} + \gamma^w_t\ket{i_k},
\end{equation}
where $U$ stands for the set of indexes of non-marked items and $M$ stands for the set of indexes of $k$ marked items.
$\alpha^w_t$, $\beta^w_t$ and $\gamma^w_t$ denote the amplitudes of respectively a non-marked item, a non-faulty marked item and the faulty marked item.

At each step of the algorithm for each of $2^t$ scenarios $w$ we shall take care of three numbers:
\begin{equation}
\label{eq:alphabetagamma}
\alpha^w_t\sqrt{N-k},\qquad\beta^w_t\sqrt{k-1}\textrm{\qquad and \qquad} \gamma^w_t.
\end{equation}

Since $\ket{\psi^w_t}$ are unit vectors we have
\begin{displaymath}
\sum_{i \in U} \left(\alpha^w_t\right)^2 + \sum_{\substack{j \in M,\\ j \neq f}} \left(\beta^w_t\right)^2 + \left(\gamma^w_t\right)^2 = 1.
\end{displaymath}
Thus, values \eqref{eq:alphabetagamma} meet the equality
\quad$\left(\vphantom{\beta^w_t\sqrt{k-1}}\alpha^w_t\sqrt{N-k}\right)^2+\left(\beta^w_t\sqrt{k-1}\right)^2+\left(\gamma^w_t\right)^2=1$\quad
and correspond to a points on the unit sphere.

Initially the mixture consists of state $\ket{\psi_0}$ with amplitudes of all items being equal,
so $\alpha=\beta=\gamma=\frac{1}{\sqrt{N}}$, and the numbers \eqref{eq:alphabetagamma} for $t=0$ are
\begin{equation}
\label{eq:alphabetagamma0}
\alpha_0\sqrt{N-k}=\frac{\sqrt{N-k}}{\sqrt{N}},\qquad\beta_0\sqrt{k}=\frac{\sqrt{k-1}}{\sqrt{N}}\textrm{\qquad and \qquad}\gamma_0=\frac{1}{\sqrt{N}}.
\end{equation}
During the first step of the algorithm
\begin{itemize}
\renewcommand{\labelitemi}{$\bullet$}
\item $D$ does not change amplitudes of the state $\ket{\psi_0}$;
\item $Q$ negates the amplitudes of all marked items: $\beta^w_1=\gamma^0_1=-\frac{1}{\sqrt{N}}$;
\item $E$ negates back the amplitude of the faulty marked item: $\gamma^1_1=-\gamma^0_1=\frac{1}{\sqrt{N}}$.
\end{itemize}

So the numbers \eqref{eq:alphabetagamma} for $t=1$ are as follows:
\begin{equation}
\label{eq:alphabetagamma1}
\begin{array}{r@{~~~}c@{~~~}c@{~~~~}c}
\alpha^0_1\sqrt{N-k}=\frac{\sqrt{N-k}}{\sqrt{N}}, & \beta^0_1\sqrt{k-1}=-\frac{\sqrt{k-1}}{\sqrt{N}}, & \gamma^0_1=-\frac{1}{\sqrt{N}} & \textrm{for $w=0$;} \\
\alpha^1_1\sqrt{N-k}=\frac{\sqrt{N-k}}{\sqrt{N}}, & \beta^1_1\sqrt{k-1}=-\frac{\sqrt{k-1}}{\sqrt{N}}, & \gamma^1_1=\frac{1}{\sqrt{N}}  & \textrm{for $w=1$.}
\end{array}
\end{equation}

According to \eqref{eq:alphabetagamma0} and \eqref{eq:alphabetagamma1},
transformation $\ket{\psi_0}\bra{\psi_0} \xrightarrow{DQ\left(E\right)} \left(1-\epsilon\right)\ket{\psi^0_1}\bra{\psi^0_1}+\epsilon\ket{\psi^1_1}\bra{\psi^1_1}$
can be represented on the unit sphere as shown on Figure~\ref{fig:fig2}.

\begin{figure}
\centering
\begin{picture}(400,320)
\put(0,0){\includegraphics[width=300pt,height=300pt]{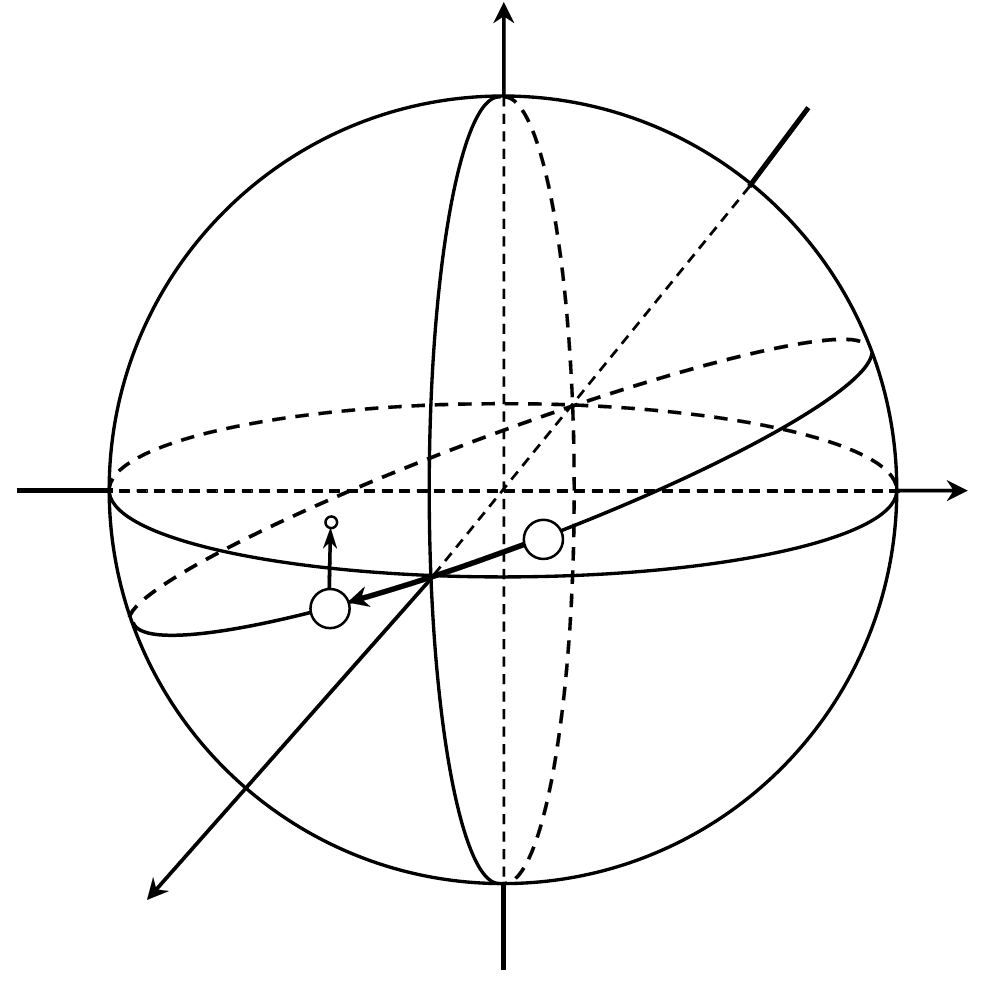}}
\put(20,142){\hbox{$-1$}}
\put(274,142){\hbox{$1$}}
\put(156,142){\hbox{$0$}}
\put(156,23){\hbox{$-1$}}
\put(156,274){\hbox{$1$}}
\put(156,274){\hbox{$1$}}
\put(22,13){\hbox{$\alpha\sqrt{N-k}$}}
\put(300,148){\hbox{$\beta\sqrt{k-1}$}}
\put(145,306){\hbox{$\gamma\sqrt{1}$}}
\put(175,132){\hbox{$\ket{\psi_0}$}}
\put(86,100){\hbox{$\ket{\psi^0_1}$}}
\put(105,140){\hbox{$\ket{\psi^1_1}$}}
\end{picture}
\caption{The first step of Grover's algorithm with a faulty marked item.
Size of a ball corresponds to a probability of the state in the mixture.
$\ket{\psi^*_0}\bra{\psi^*_0} = 1 \ket{\psi_0}$ and
$\ket{\psi^*_1}\bra{\psi^*_1} = \left(1-\epsilon\right)\ket{\psi^0_1}\bra{\psi^0_1}+\epsilon\ket{\psi^1_1}\bra{\psi^1_1}$
}
\label{fig:fig2}
\end{figure}
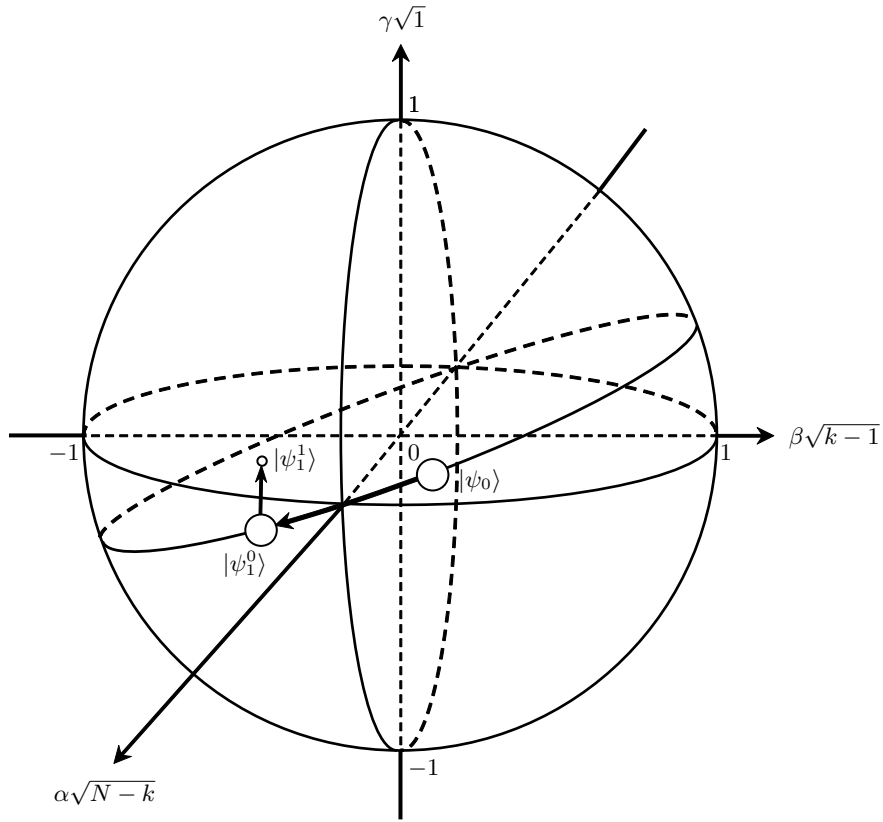

Note that if $\epsilon=0$ the state of the algorithm travels clockwise along the slanted orthodrome\footnote{Orthodrome, also known as a great circle, of a sphere is the intersection of the sphere with a plane which passes through the center point of the sphere.} which contains points $\ket{\psi_0}$ and $\left(1,0,0\right)$.
The travel lasts until the state reaches the neighborhood of the vertical orthodrome (hereafter we call it \textit{meridian}) where $\alpha=0$
after $t \approx 0.785\ldots\sqrt{N/k}$ steps.

\begin{figure}[h]
\centering
\begin{picture}(150,150)
\put(25,0){\includegraphics[width=150pt,height=150pt]{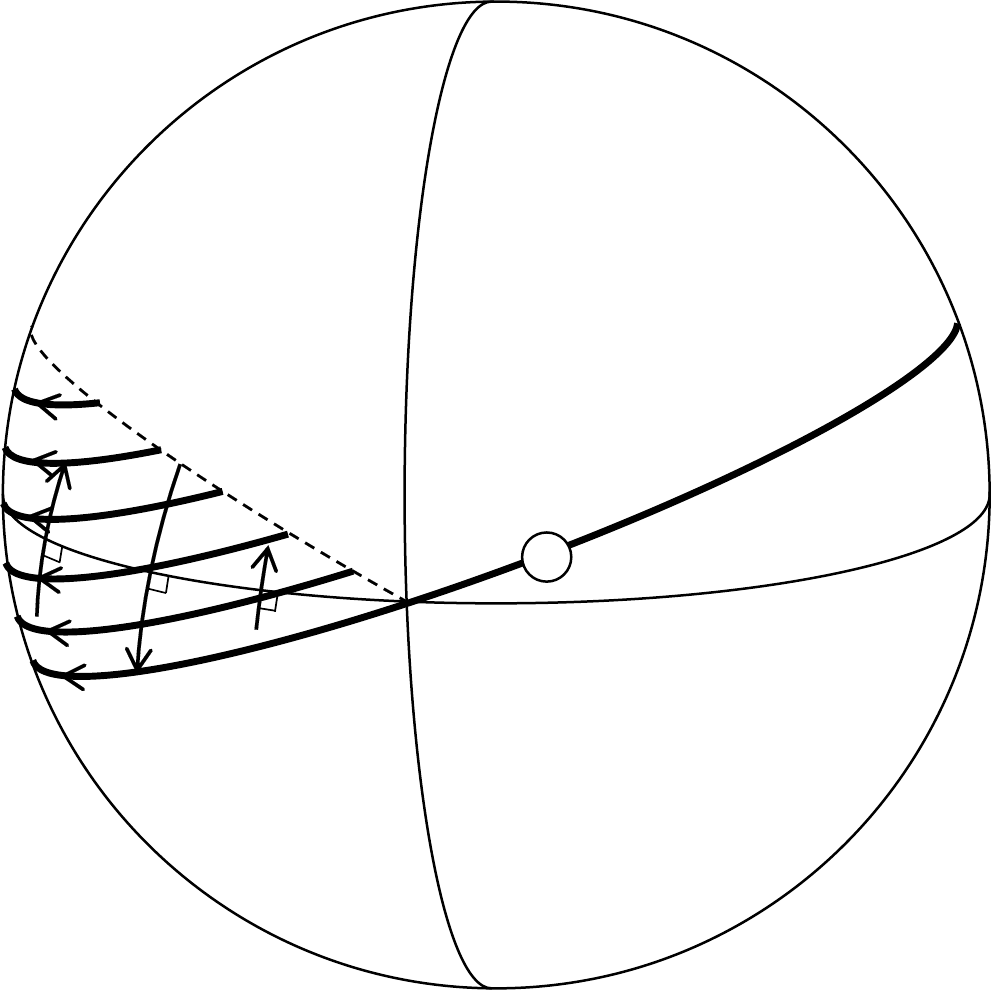}}
\put(100,73){\hbox{$\ket{\psi_0}$}}
\end{picture}
\caption{Probable routes of the initial state $\ket{\psi_0}$ in the run of Grover's algorithm with a faulty marked item.}
\label{fig:fig3}
\end{figure}

If $k \ll N$ the rotation angle \eqref{eq:rotation} is sufficiently small. The run of the algorithm can be viewed as
clockwise rotation of a state in parallel to the slanted orthodrome ($DQ$-movement)
with occasional ($\epsilon$-probable) up-and-down jumps around the horizontal orthodrome (hereafter we call it \textit{equator}),
which correspond to operator $E$, as shown on Figure~\ref{fig:fig3}.
During the first $0.785\ldots\sqrt{N/k}$ steps, the state cannot go out of the area which is covered with arrows on the figure.

As we already mentioned, $0.785\ldots\sqrt{N/k}$ steps are necessary to reach the desired plane where $\alpha=0$,
given that no fault occurs on the way (in our notation: in expression \eqref{eq:psiwt},
$\alpha^{00\ldots 0}_{\lfloor 0.785\ldots\sqrt{N/k}\rfloor}\approx 0$).

But what could the length of the route be if some faults occur on the state's way to the desired plane?
In the two following subsections we will derive upper bounds for the effect of these faults.

\subsubsection{At most one fault}

First, let us assume that the total number of faults is at most one.
Although this assumption seems to be rather implausible, we have some arguments for it:
\begin{itemize}
\item for sufficiently small $\epsilon$, we have $\ldots \ll \epsilon^3 \ll \epsilon^2 \ll \epsilon$,
so that probability of more than one fault
$\sum_{f=2}^{t}\left(\substack{t\\ f}\right) \epsilon^f\left(1-\epsilon\right)^{t-f}<t^2\epsilon^2$
could be neglected for number of steps $t \in o\left(\frac{1}{\epsilon}\right)$;
\item as we shall see later, the second and all subsequent faults have smaller effect
and even have great chances to drive the state closer to the desired meridian.
\end{itemize}

\begin{figure}[h]
\centering
\begin{picture}(150,150)
\put(25,0){\includegraphics[width=150pt,height=150pt]{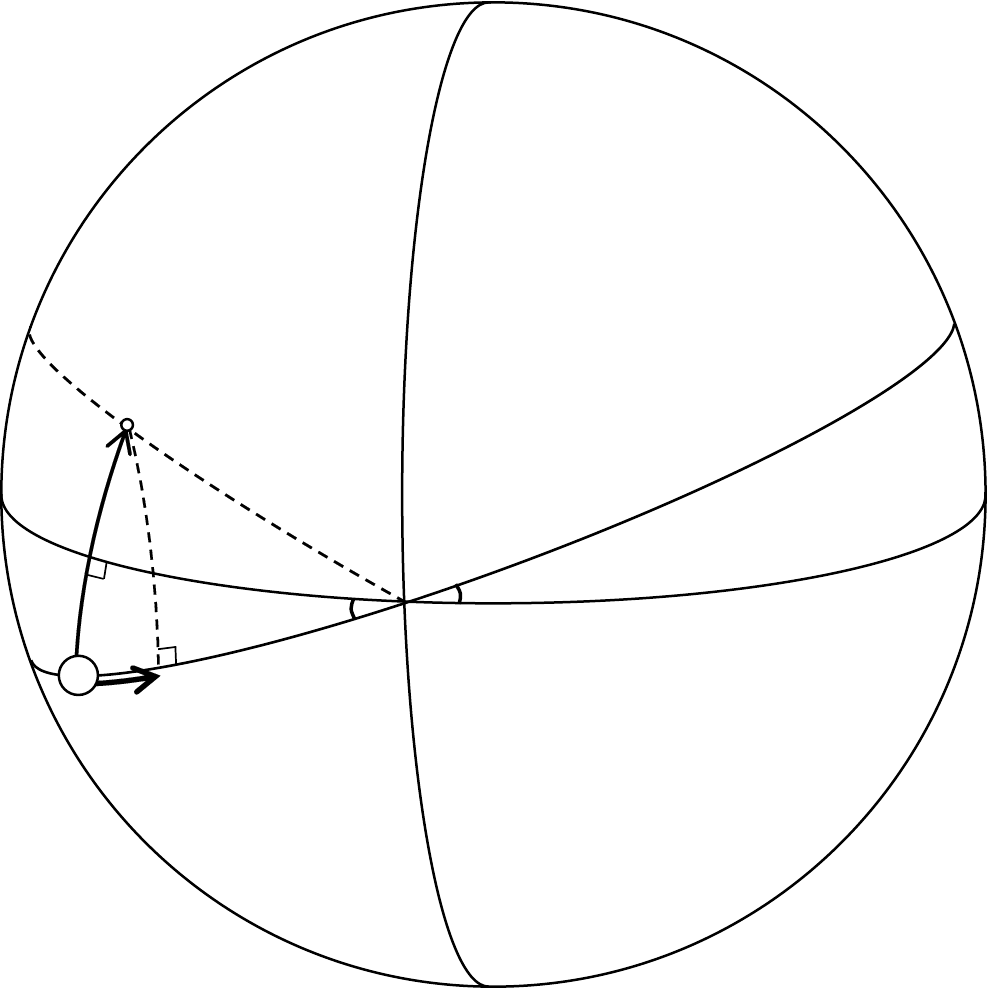}}
\put(10,33){\colorbox{white}{$\ket{\psi^{\ldots 0}_t}$}}
\put(43,90){\hbox{$\ket{\psi^{\ldots 1}_t}$}}
\put(30,55){\hbox{$a$}}
\put(40,40){\hbox{$c_{\rm err}$}}
\put(65,54){\hbox{$A$}}
\end{picture}
\caption{Metrics for a fault.}
\label{fig:fig4}
\end{figure}
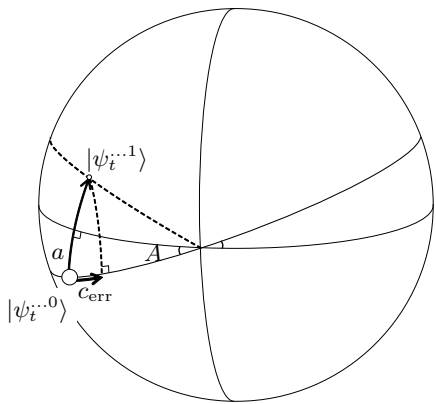

Let us calculate the effect of a fault in the sense of its projection onto the ``no-faults'' route.
On the Figure~\ref{fig:fig4} we illustrate $\epsilon$-probable transformation
$\ket{\psi^{\ldots 0}_t} \xrightarrow{E} \ket{\psi^{\ldots 1}_t}$, which happened on some step $t$.
The fault increases the angular distance to the desired meridian ($\alpha=0$) by $c_{\rm err}$.
Using rules of spherical trigonometry (\eqref{eq:R8} and \eqref{eq:R7}) we have:

\begin{equation}
\label{eq:erroreffect}
c_{\rm err} = \arctan{\left(\tan{2a}~\cos\arcsin{\frac{\cos{A}}{\cos{a}}}\right)} = \arctan{\left(\tan{2a}~\sqrt{1-\frac{\cos^2{A}}{\cos^2{a}}}\right)},
\end{equation}
where $A = \arctan{\frac{1}{\sqrt{k-1}}}$ is angle between the two equators, and $a$ is the distance between $\ket{\psi^{\ldots 0}_t}$ and the horizontal equator.
Note that $a$ is at most $A$ ($a=A$ only when $\ket{\psi^w_t}$ reaches the desired meridian), so $1-\frac{\cos^2{A}}{\cos^2{a}} \ge 0$.

In equation \eqref{eq:erroreffect} we assumed that $\ket{\psi^{\ldots 0}_t}$ is located on the bottom-margin of the arrow-filled area of the Figure~\ref{fig:fig3} (which always holds for one fault case).
For $\ket{\psi^{\ldots 0}_t}$ located above the bottom margin, we should calculate \eqref{eq:erroreffect} for a smaller value of $A$,
which will result in smaller value of $c_{\rm err}$.
For $\ket{\psi^{\ldots 0}_t}$ located above the horizontal equator, the fault-effect $c_{\rm err}$ is negative,
i.e. the resulting state $\ket{\psi^{\ldots 1}_t}$ is closer to the desired meridian
w.r.t. the direction in parallel to the slanted orthodrome.
Relaxing the above-mentioned assumption, we can conclude the following rough bound:

\begin{equation}
\label{eq:erroreffectbound0}
c_{\rm err} \le \arctan{\left(\tan{2a}~\sqrt{1-\frac{\cos^2{A}}{\cos^2{a}}}\right)}.
\end{equation}

Now, based on the inequality \eqref{eq:erroreffectbound0} we shall derive more precise bounds.

If the number of non-faulty marked items $k-1 \ge 1$, then the angle
$A = \arctan{\frac{1}{\sqrt{k-1}}} \in \left(\arctan{\frac{1}{\infty}};\arctan{\frac{1}{1}}\right] = \left(0;\frac{\pi}{4}\right]$.
So we have $0 \le a \le A \le \frac{\pi}{4}$.

Now let us consider different values of $a$. If $0 \le a \le \frac{\pi}{6}$, then \eqref{eq:erroreffect} is bounded by
\begin{equation}
\label{eq:erroreffectbound1}
c_{\rm err} \le \max_{\substack{0 \le a \le A,\\ a \le \frac{\pi}{6}}}{\arctan{\left(\tan{2a}~\sqrt{1-\frac{\cos^2{A}}{\cos^2{a}}}\right)}} \le \frac{\pi}{4},
\end{equation}
where the inequalities become equalities for $A=\frac{\pi}{4}$ and $a=\frac{\pi}{6}$.

If $\frac{\pi}{6} < a \le \frac{\pi}{4}$, then we can follow that the state $\ket{\psi^{\ldots 0}_t}$
is gone far away from the point ``$\alpha\sqrt{N-k}=1$'' of the unit sphere.
This distance between the point $\left(1,0,0\right)$ and the state $\ket{\psi^{\ldots 0}_t}$ can be derived from the rule \eqref{eq:R2}:
\begin{equation}
\label{eq:erroreffectbound2}
c = \arcsin{\frac{\sin{a}}{\sin{A}}} \ge \arcsin{\frac{\sin{\pi/6}}{\sin{\pi/4}}} = \frac{\pi}{4}
\end{equation}

Since the total distance between the point ``$\alpha\sqrt{N-k}=1$'' and any point of the meridian ``$\alpha\sqrt{N-k}=0$'' is exactly $\frac{\pi}{2}$,
we follow that the state $\ket{\psi^{\ldots 0}_t}$ is at most $\frac{\pi}{2}-\frac{\pi}{4}=\frac{\pi}{4}$ far from the desired meridian.

From \eqref{eq:erroreffectbound1} and \eqref{eq:erroreffectbound2} we formulate the following joint conclusion:
\begin{corollary}
\label{cor:onefault}
At least one of the following claims holds for any state $\ket{\psi^{\ldots 0}_t}$ with $0 \le a \le A \le \frac{\pi}{4}$:
\begin{itemize}
\item either the fault-effect $c_{\rm err} \le \frac{\pi}{4}$,
\item or the state $\ket{\psi^{\ldots 0}_t}$ is already at most $\frac{\pi}{4}$ far from the desired meridian ``$\alpha=0$''.
\end{itemize}
\end{corollary}

\subsubsection{Any number of faults}

Now let us use another approach to study the evolution of a state in the considered settings.
Transformation $\ket{\psi^w_t} \xrightarrow{DQ} \ket{\psi^{w,0}_{t+1}}$ drives the state $\ket{\psi^w_t}$
clockwise in parallel to the slanted orthodrome by a distance, which depends on the position of this state on the unit sphere.
\\

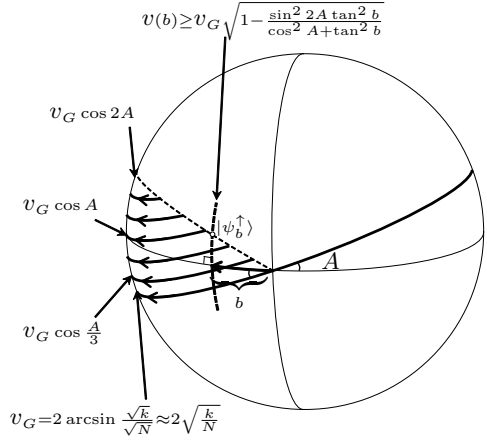
\begin{figure}[h]
\begin{picture}(150,150)
\put(55,0){\includegraphics[width=150pt,height=150pt]{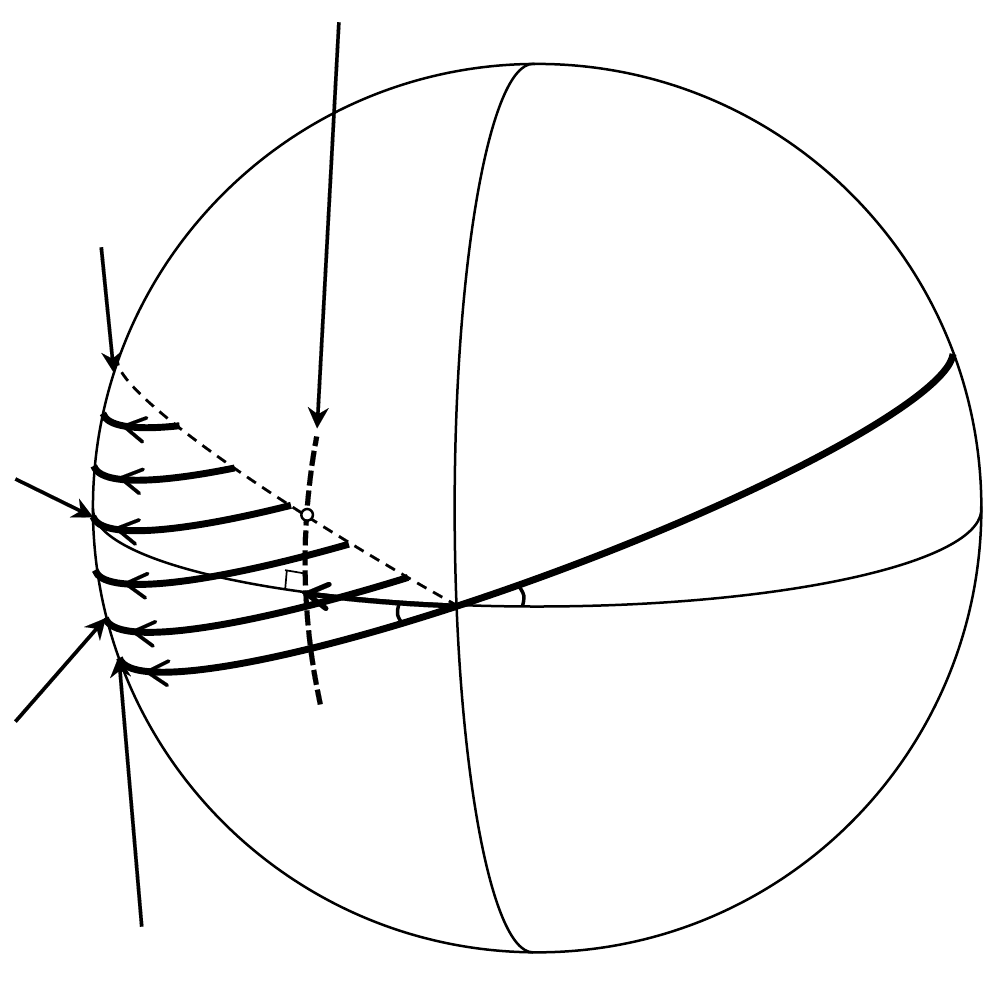}}
\put(101,57){\hbox{$\underbrace{~~~~~~~}_b$}}
\put(143,59){\hbox{$A$}}
\put(25,0){\hbox{$v{\scriptstyle_G = 2\arcsin{\frac{\sqrt{k}}{\sqrt{N}}} \\ \approx 2\sqrt{\frac{k}{N}}}$}}
\put(30,33){\hbox{$v{\scriptstyle_G \cos{\frac{A}{3}}}$}}
\put(30,81){\hbox{$v{\scriptstyle_G \cos{A}}$}}
\put(40,116){\hbox{$v{\scriptstyle_G \cos{2A}}$}}
\put(75,151){\hbox{$v{\scriptstyle\left(b\right)\ge}v{\scriptstyle_G \sqrt{1-\frac{\sin^2{2A}\tan^2{b}}{\cos^2{A}+\tan^2{b}}}}$}}
\put(103,73){\hbox{${\scriptstyle\ket{\psi^\uparrow_b}}$}}
\end{picture}
\caption{``Speed'' of a state located on some meridian $b$.}
\label{fig:fig5}
\end{figure}

On Figure~\ref{fig:fig5} we show a ``speed'' for each possible position of the state $\ket{\psi^w_t}$.
On the ``no-faults'' route (i.e. the slanted orthodrome) the speed coincides with that of the original Grover's algorithm:
\sloppy{$v_{\scriptstyle_G}~=~2\arcsin{\frac{\sqrt{k}}{\sqrt{N}}}~\approx~2\sqrt{\frac{k}{N}}$}.
After a state jumps up, its speed decreases depending on its distance to the slanted orthodrome:
e.g. on the parallel circle which contains point ``$\beta\sqrt{k-1}=1$'' of the units sphere,
its speed is $v_{\scriptstyle_G} \cos{A}$.

We can also calculate the speed of a state w.r.t. the direction in parallel to the horizontal equator,
i.e. the speed of its projection onto the horizontal equator $v_{\scriptstyle_\Pi}$.
Obviously, among all states on some meridian $b$, the uppermost state $\ket{\psi^\uparrow_b}$ has the least speed.
Distance between $\ket{\psi^\uparrow_b}$ and the point ``$\alpha\sqrt{N}=1$'' of the unit sphere,
can be derived from angle $A$, distance $b$ and rule \eqref{eq:R6}: $c = \arctan{\frac{\tan{b}}{\cos{A}}}$.
Distance between $\ket{\psi^\uparrow_b}$ and the slanted orthodrome
can be derived from distance $c$, angle $2A$ and rule \eqref{eq:R2}:
\sloppy{$a^\prime~=~\arcsin{\left(\sin{2A}\sin{c}\right)}~=~\arcsin{\left(\sin{2A}\sin{\arctan{\frac{\tan{b}}{\cos{A}}}}\right)}$}.
And the speed of $\ket{\psi^\uparrow_b}$
\begin{displaymath}
v^\uparrow\left(b\right)
= v_{\scriptstyle_G}\cos{\left(a^\prime\right)}
= v_{\scriptstyle_G}\sqrt{1-\frac{\sin^2{2A}~\tan^2{b}}{\cos^2{A}+\tan^2{b}}}
\end{displaymath}
serves as a natural lower bound for the speed of any state located on the same meridian $b$:
\begin{displaymath}
v\left(b\right)
\ge v^\uparrow\left(b\right)
= v_{\scriptstyle_G}\sqrt{1-\frac{\sin^2{2A}~\tan^2{b}}{\cos^2{A}+\tan^2{b}}}.
\end{displaymath}
Projection of the speed $v\left(b\right)$ onto the horizontal equator might be slightly less:
\begin{equation}
\label{eq:pspeed}
v_{\scriptstyle_\Pi}\left(b\right)
\ge v\left(b\right)\cos{A}
\ge v_{\scriptstyle_G}\sqrt{1-\frac{\sin^2{2A}~\tan^2{b}}{\cos^2{A}+\tan^2{b}}}\cos{A}.
\end{equation}
Vertical jumps leave states on the same meridians, so faults does not affect value $b$.
Moving at least at speed \eqref{eq:pspeed}, a state will pass distance $\frac{\pi}{2}$ in at most
$\int_0^\frac{\pi}{2} \! \frac{1}{v_{\scriptstyle_\Pi}\left(b\right)} \, \mathrm{d}b$ steps.
From \eqref{eq:pspeed} we obtain an upper bound for the number of steps until some meridian $b^*$
(for arbitrary number of faults, i.e. for any $\epsilon$):

\begin{equation}
\label{eq:steps}
t_{b^*}
\le \int_0^{b^*} \! \frac{1}{v_{\scriptstyle_\Pi}\left(b\right)} \, \mathrm{d}b
\le \int_0^{b^*} \! \frac{1}{v_{\scriptstyle_G}\sqrt{1-\frac{\sin^2{2A}~\tan^2{b}}{\cos^2{A}+\tan^2{b}}}\cos{A}} \, \mathrm{d}b.
\end{equation}

We note that on the fastest ``no-fault'' route a state will travel exactly at speed $v_{\scriptstyle_G}$,
so that in the same many $t_{b^*}$ steps it can reach at most $\left(v_{\scriptstyle_G}t_{b^*}\right)^{\rm th}$ meridian.

From \eqref{eq:steps} we can derive the upper bound for the distance between the two meridians $b^*$ and $v_{\scriptstyle_G}t_{b^*}$:
\begin{equation}
\label{eq:difference}
v_{\scriptstyle_G}t_{b^*} - b^*
\le \xcancel{v_{\scriptstyle_G}}\int_0^{b^*} \! \frac{1}{\xcancel{v_{\scriptstyle_G}}\sqrt{1-\frac{\sin^2{2A}~\tan^2{b}}{\cos^2{A}+\tan^2{b}}}\cos{A}} \, \mathrm{d}b - b^*
\end{equation}

For example, if we fix $b^* = \frac{3}{8}\pi$, then for any $0 \le A \le 0.1953\ldots\pi$, value \eqref{eq:difference}
does not exceed $\frac{\pi}{4}$. That is, when the fastest state of a quantum ensemble reaches meridian
$b^* + \frac{\pi}{4} = \frac{5}{8}\pi$,
the slowest state of the ensemble with certainty reaches at least meridian $b^* = \frac{3}{8}\pi$.

\subsection{Proof of theorem \ref{thm:search}}
We run standard Grover's algorithm $\frac{5}{4}$ times longer than usually, and then perform a measurement.
\begin{itemize}
\renewcommand{\labelitemi}{$\bullet$}
\item If at most one fault is promised, and $A \le 0.25\pi$,\footnote{$A \le 0.25\pi$ means that there is at least as many non-faulty marked items as faulty marked items. Since we limit our considerations with only one faulty marked item, it suffices with only one non-faulty marked item.} then we use Corollary \ref{cor:onefault} and follow that any component of the resulting mixture $\ket{\psi^*}$ is at most $\frac{\pi}{8}$ far from the meridian ``$\alpha=0$''.
\item If $A \le 0.1953\ldots\pi$,\footnote{For one faulty marked item, it means existence of at least $\lceil \arccot{0.1953\ldots\pi} \rceil = \lceil 1.02047\rceil=2$ non-faulty marked items} then we substitute $b^* = \frac{3}{8}\pi$ in \eqref{eq:difference} and follow exactly the same.

\end{itemize}
Measurement of such $\ket{\psi^*}$ results in finding a marked item (the faulty or a non-faulty one) with probability at least $\cos^{2}\frac{\pi}{8} = 0.853553\ldots$
\begin{itemize}
\renewcommand{\labelitemi}{$\bullet$}
\item Otherwise, if $0.1953\ldots\pi < A \le 0.25\pi$ and there is no promise on the number of faults, it means that there is exactly 1 non-faulty marked element (so $A=0.25\pi$).
\end{itemize}
In this specific case we run standard Grover's algorithm $\approx 1.34$ times longer than usually, and then perform a measurement.
We substitute $A=0.25\pi, b^* \approx 0.33\pi$ in \eqref{eq:difference} and follow that value \eqref{eq:difference} does not exceed $\approx 0.34\pi$.
That is, when the fastest state of a quantum ensemble reaches meridian $b^* + 0.34\pi = 0.67\pi (=0.5 \times 1.34\pi)$, the slowest state of the ensemble with certainty reaches at least meridian $b^* = 0.33\pi$.
Measurement of such $\ket{\psi^*}$ results in finding a marked item (the faulty or a non-faulty one) with probability at least $\cos^{2}{0.17\pi} = 0.74\ldots$.
\qed


\section{Limiting behavior}

In this section we analyze the state of the algorithm after a large number of steps. We consider a density matrix of the state and describe how transformations of the algorithm change the state. We prove an existence of the limiting state and give probabilities of finding a non-faulty marked, as well as faulty marked item then measuring the state.

\subsection*{Proof of theorem \ref{thm:limiting_behaviour}}

Consider the density matrix $\rho_t$ of the quantum state of Grover's algorithm after $t$ queries. 
Due to symmetry, we can assume that the first $k$ basis states correspond to the marked elements. Note that Grover's algorithm acts in the same way on all non-faulty marked elements, as well as on all non-marked elements.
Therefore, the state of the algorithm is a probabilistic mixture of pure states of the form
\begin{equation}
\label{eq:state} 
\alpha \ket{\psi_{+}} + \beta \ket{i_k} + \gamma \ket{\psi_{-}} .
\end{equation}
\
The density matrix $\rho_t$, then, takes the form
\ 
\begin{equation}
\label{eq:rho_t}
\rho_t= \left[ 
\begin{array}{ccccccc}
   a   & \ldots &    a   &    a'  &    c   & \ldots &    c   \\
\vdots & \ddots & \vdots & \vdots & \vdots & \ddots & \vdots \\
   a   & \ldots &    a   &    a'  &    c   & \ldots &    c   \\
   a'  & \ldots &    a'  &    b   &    d'  & \ldots &    d'  \\
   c   & \ldots &    c   &    d'  &    d   & \ldots &    d   \\
\vdots & \ddots & \vdots & \vdots & \vdots & \ddots & \vdots \\
   c   & \ldots &    c   &    d'  &    d   & \ldots &    d
\end{array} 
\right]
\end{equation}

\noindent
because the density matrix for every pure state (\ref{eq:state}) in the mixture $\rho_t$ is of this form. 

Next we consider the effect of the faulty query $Q'$ and the diffusion transformation $D$ on the density matrix $\rho$.


\begin{Lemma}
The effect of the faulty query $Q'$ on the density matrix $\rho_t$ is:

\begin{equation}
\label{eq:fauly_query}
\begin{array}{lcl}
a  & \mapsto & a \\
a' & \mapsto & -(2p - 1) a' \\
b  & \mapsto & b \\
c  & \mapsto & -c \\
d' & \mapsto & (2p - 1) d' \\
d  & \mapsto & d
\end{array} .
\end{equation}

\end{Lemma}

\begin{proof}
Follows from \cite[formula (2)]{AB+13}.
\end{proof}


\begin{Lemma}
The effect of the diffusion transformation $D$ on the density matrix $\rho$ is:

\begin{equation}
\label{eq:diffusion}
\rho_{i,j} \mapsto 4V - 2V_{i,\cdot} - 2 V_{\cdot, j} + \rho_{i,j} ,
\end{equation}
where $V = \frac{1}{N^2} \sum_{i,j} \rho_{i,j} $ is the average of all elements of $\rho$, $V_{i,\cdot} = \frac{1}{N} \sum_{j} \rho_{i,j} $ is the average of $i^{th}$ row and $V_{\cdot,j} = \frac{1}{N} \sum_{i} \rho_{i,j} $ is the average of $j^{th}$ column of $\rho$. 
\end{Lemma}

\begin{proof}
First we right multiply $\rho = (\rho_{i,j})$ by 
$$ 
D =
\left[
\begin{array}{cccc} 
-1 + \frac{2}{N} & \frac{2}{N} & \ldots & \frac{2}{N} \\
\frac{2}{N} & -1 + \frac{2}{N} & \ldots & \frac{2}{N} \\
\ldots & \ldots & \ldots & \ldots \\
\frac{2}{N} & \frac{2}{N} & \ldots & -1 + \frac{2}{N} 
\end{array} 
\right]
.
$$
\
We get matrix $\rho' = \rho D$ where
\
\begin{equation}
\label{eq:rho_D}
\rho'_{i,j} = \sum_k \rho_{i,k} \frac{2}{N} - \rho_{i,j} =
2V_{i,\cdot} - \rho_{i,j} .
\end{equation}
\
Next we left multiply \eqref{eq:rho_D} by $D$ and get matrix $\rho'' = D \rho D$ where
$$
\rho''_{i,j} = \sum_{k} \frac{2}{N} (2V_{k,\cdot} - \rho_{k,j}) - (2V_{i,\cdot} - \rho_{i,j}) =
$$
$$
= \frac{4}{N} \sum_{k} V_{k,\cdot} - \frac{2}{N} \sum_{k} \rho_{k,j} - 2V_{i,\cdot} + \rho_{i,j} =
$$
$$
4V - 2V_{i,\cdot} - 2V_{\cdot,j} + \rho_{i,j} .
$$
\

\end{proof}


\noindent
As $V_{i,\cdot} = V_{\cdot,i}$ (this follows from the structure of the density matrix $\rho_t$) we can use $V_i$ to denote both $V_{i,\cdot}$ and $V_{\cdot,i}$.

Using the same argument as in \cite{AB+13} we can prove 
$\lim_{t \rightarrow \infty} a' = 0$ and $\lim_{t \rightarrow \infty} d' = 0$.
Consider $t'$ such that for all $t > t'$ inequalities $a' < \epsilon$ and $d' < \epsilon$ hold.
We will prove that at this moment all elements of $\rho_t$ should be $O(\epsilon)$ close to their final values. The proof is done in two steps. First we prove that all $V_i$ can differ by at most $O(\epsilon)$. Next, as a consequence, we derive $\Delta \rho_{i,j} = O(\epsilon)$.


\begin{Claim}
\label{claim:Vi_Vj}
$\forall i,j: V_i - V_j = O(\epsilon)$.
\end{Claim}

\begin{proof}

First, note that $V_1 = \ldots = V_{k-1}$ and $V_{N-k} = \ldots = V_N$. 
This follows from the structure of $\rho_t$.

Next, consider effect of $D$ on elements $\rho_{1,k}$ and $\rho_{k,N}$.
Both elements are $O(\epsilon)$.
Corresponding changes $\Delta \rho_{1,k} = 4V - 2V_1 - 2V_k$ and $\Delta \rho_{k,N} = 4V - 2V_k - 2V_N$ should also be $O(\epsilon)$.
Thus, 
\begin{equation}
\label{eq:V1_VN}
\Delta \rho_{1,k} - \Delta \rho_{k,N} = 2V_1 - 2V_N = O(\epsilon)
\end{equation}
from which follows $V_1 - V_N = O(\epsilon)$.

Now consider $\Delta \rho_{1,k} = 4V - 2V_1 - 2V_k = O(\epsilon)$.
By using \eqref{eq:V1_VN} we can write
$$
V = \frac{1}{N} ((k-1) V_1 + V_k + (N-k) V_N) = \frac{1}{N} ((N-1) V_1 + V_k + O(\epsilon)) . 
$$
Putting this into $\Delta \rho_{1,k}$ and opening brackets we get
$$
\frac{2N - 4}{N} V_1 - \frac{2N - 4}{N} V_k = O(\epsilon)
$$
from which follows $ V_1 - V_k = O(\epsilon) $.

\end{proof}


\begin{corollary}
$D$ changes elements of $\rho_t$ by at most $O(\epsilon)$.
\end{corollary}

\begin{proof}
The change of element $\rho_{i,j}$ is equal to $\Delta \rho_{i,j} = 4V - 2V_i - 2V_j$. 
Using claim \ref{claim:Vi_Vj} we can rewrite this equation as
$$
\Delta \rho_{i,j} = \frac{4}{N}(N V_i + O(\epsilon)) - 4V_i + O(\epsilon) = O(\epsilon) .
$$

\end{proof}


\begin{corollary}
$c = O(\epsilon)$.
\end{corollary}

\begin{proof}
Consider the effect of $Q'$ on $V_1 = (k-1)a + (N-k)c + O(\epsilon)$.
$Q'$ leaves $a$ unchanged, but changes the sign of $c$. 
As the claim \ref{claim:Vi_Vj} should hold for all subsequent steps, $Q'$ can not change value of $V_1$ by more than $O(\epsilon)$.
Thus, $V_1(c) - V_1(-c)$ should be $O(\epsilon)$, which means that $c = O(\epsilon)$.

\end{proof}


\begin{corollary}
$Q'$ changes elements of $\rho_t$ by at most $O(\epsilon)$.
\end{corollary}

\begin{proof}
Trivially follows from \eqref{eq:fauly_query} and the previous corollary.
\end{proof}


We can chose $\epsilon$ to be arbitrary small, thus, we have
\ 
$$
\rho_{lim} = 
\left[ 
\begin{array}{ccccccc}
   a   & \ldots &    a   &    0   &    0   & \ldots &    0   \\
\vdots & \ddots & \vdots & \vdots & \vdots & \ddots & \vdots \\
   a   & \ldots &    a   &    0   &    0   & \ldots &    0   \\
   0   & \ldots &    0   &    b   &    0   & \ldots &    0   \\
   0   & \ldots &    0   &    0   &    d   & \ldots &    d   \\
\vdots & \ddots & \vdots & \vdots & \vdots & \ddots & \vdots \\
   0   & \ldots &    0   &    0   &    d   & \ldots &    d
\end{array} 
\right]
.
$$
As $(k-1)a = b = (N-k)d$ (follows from $V_1 = V_k = V_N$) and $(k-1)a + b + (N-k)d = 1$ (property of density matrix) we have $(k-1)a = b = (N-k)d = \frac{1}{3}$ which completes the proof of the theorem.

\qed


\section{Summary and open problems}

In this paper we focus on the case where search space contains multiple non-faulty and one faulty marked element. 

First we show that if there are at least two non-faulty marked elements or it there is at most one fault, then it is still possible to find a non-faulty marked element in $O(\sqrt{N})$ queries with $\Theta(1)$ probability.

Second, we analyze the limiting behavior of the algorithm for a large number of steps and show the existence and structure of limiting state $\rho_{lim}$.

It is an open question to generalize the results for more than one faulty marked element. Although the generalization of Theorem \ref{thm:limiting_behaviour} seems to be straightforward it is not given here. The generalization of Theorem \ref{thm:search} might be tricky as one needs to deal with hyper-spherical geometry.
One simple case, in which the theorem still holds for multiple faulty marked elements, is synchronized-fault-case, where each query either fails or succeeds on all faulty marked elements.



\end{document}